\documentclass[a4paper,12pt]{article}
\pdfoutput=1
\usepackage{amsfonts}
\usepackage{amsbsy}
\usepackage{amssymb}
\usepackage{amsmath}
\usepackage{amsthm}
\usepackage{mathtools, nccmath}
\usepackage{graphicx} 
\usepackage{comment}
\usepackage{xcolor}

\newtheorem{remark}{Remark}
\newtheorem{proposition}{Proposition}

\title{Analysis of the Effects of Curvature on the Solutions
  of Shallow Water Equations}

\author{Stelian Ion$^*$, Dorin Marinescu$^*$, Stefan-Gicu
  Cruceanu\footnote{{``Gheorghe Mihoc - Caius Iacob''
      Institute of Mathematical Statistics and Applied
      Mathematics, Romanian Academy, 050711 Bucharest,
      Romania, {\tt emails: ro\_diff@yahoo.com,
        marinescu.dorin@ismma.ro, stefan.cruceanu@ismma.ro}.}}}

\date{}

\begin{document}
\maketitle

\begin{abstract}
  The most used form of the Shallow Water Equations doesn't
  take into account the variation of the curvature of the
  base flow surface.  In this paper, we compare the
  theoretical and numerical solutions of the standard model
  and with the solutions of an extended model and for a
  large class of base flow surfaces.  We find that the
  solution of the standard model is still a good
  approximation for the extended model for many real life
  applications.
\end{abstract}

\section{Introduction}
For a large number of applications, especially for the ones
in hydrology
\cite{bouchut-book,COZZOLINO201883,dooge,Hsu2019}, the
mathematical model of the water flow on the soil surface is
given by (S-SWE)
\begin{equation}
  \begin{split}
    \displaystyle\frac{\partial h}{\partial t}+\partial_j\left(h v^j\right) &=\mathfrak{M},\\
    \displaystyle\frac{\partial (h v_i)}{\partial t}+\partial_j\left(h v_iv^j\right)+h\partial_iw &=\mathfrak{F}_i, \quad i=1,2,
  \end{split}
  \label{s-swe} 
\end{equation}
where $h(t,\boldsymbol{x})$ is the water depth and
$v_i(t,\boldsymbol{x})$ and $v^i(t,\boldsymbol{x})$ are the
covariant and contravariant components of the water velocity
$\boldsymbol{v}(t,\boldsymbol{x})$.  We used here the
Einstein summation convention for writing the model i
n this
compact form.  The potential of the water level is
\begin{equation*}
  w = g \left[ z(\boldsymbol{x}) + h(t,\boldsymbol{x}) \right],
\end{equation*}
where $z(\boldsymbol{x})$ is the altitude of the soil
surface and $g$ represents the gravitational acceleration.
The contribution of rain and infiltration to the water mass
balance is taken into consideration through $\mathfrak{M}$,
while $\mathfrak{F}_i$ quantify the rates of momentum
production (they can include plant cover resistance and
fluid-soil friction, for example).  

\section{Shallow Water Equations for Surface with Variable Curvature, E-SWE}
If the base surface of the flow exhibits different
orientations or moderate variations of the geometrical
properties of its metrics, then it is convenient (at least
from a theoretical point of view) to work with a surface
base coordinate system.  Let us assume that the base flow
surface admits a parametric representation of the form
\begin{equation*}
  x^i=b^i(y^1,y^2), \quad i=\overline{1,3}, \quad \left(y^1,y^2\right)\in D\subset\mathbb{R}^2.
\end{equation*}
Using this representation of base flow surface, one
introduces a coordinate system in $\mathbb{R}^3$ by
\begin{equation*}
  x^i=b^i(y^1,y^2)+y^3\nu^i, \quad y^3\in(0,\varepsilon),
\end{equation*}
where $\boldsymbol{\nu}$ is the unitary normal vector to the
surface and gravitational upward oriented.  In this new
coordinate system, using a hydrostatic approximation of the
pressure field, the shallow water equations with curvature
E-SWE read as \cite{sds-ffonhv-arxiv}
\begin{equation}
  \begin{split}
    \displaystyle\frac{\partial (\beta h)}{\partial t} + \partial_a\left(\beta hv^a\right)
    &=\beta\mathfrak{M},\\
    \displaystyle\frac{\partial (h\beta v^c)}{\partial t} +
    \partial_a\left(\beta h v^c v^a\right) + h\beta\gamma^{c}_{ab} v^a v^b + 
    h\beta\beta^{ca}\partial_a{w}
    &= \beta\mathfrak{F}^c,
  \end{split}
  \label{e-swe}
\end{equation}
where $\{v^a\}_{a=1,2}$ are the contravariant components of
the velocity field $\boldsymbol{v}$, $\beta$ is the area
element of the surface, $\{\beta^{ab}\}_{a,b=1,2}$ are the
contravariant components of the metric tensor of the
surface, and $w=g\left[b^3(y^1,y^2)+h\nu^3(y^1,y^2)\right]$.
Christoffel symbols $\gamma_{\cdot\, \cdot}^{\cdot}$ and the
normal vector $\boldsymbol{\nu}$ intermediate the influence
of the surface curvature on the solution of the mathematical
model.  It is assumed that the values of the surface
curvature are not so high to preclude the hydrostatic
approximation of the pressure of the models.  For surfaces
with high values of curvature, the hydrostatic approximation
of the pressure is coarse and the models must include the
curvature tensor of the surface. For such cases, one can see
\cite{dressler, dewals, Zerihun, bergercarey}.

The shallow water equations E-SWE extend the equations
\mbox{S-SWE} (also known as standard shallow water
equations) which are commonly used from cases of almost flat
surfaces to cases of arbitrary surfaces.  The \mbox{E-SWE}
is an intermediate model between \mbox{S-SWE} and the ones
using the curvature tensor of the surface.  In the
\mbox{S-SWE} model, the variation of the geometrical
characteristics of the surface is taken into account only
through the gradient of the surface.  In addition to the
gradient of the surface, the E-SWE model also includes the
gradient of the unitary normal vector to the surface which
is related to the curvature of the surface.

It is important to know the effects of the curvature of the
base flow surface on the solutions of the two models.  For
the S-SWE model, the variation of the geometrical
characteristics of the surface is taken into account only
through the gradient of the surface and since the E-SWE
model explicitly includes terms related to the curvature of
the surface, the following question rises: is the gradient
of the surface sufficient to catch the main effect of the
curvature on the flow dynamics or does one need new terms
for the model to increase its accuracy?  It is not easy to
give a complete answer to this question, but we try here to
partially answer for a class of surfaces that are of
interest in practical applications.  We analyze the
theoretical solutions of the two models for the case of
steady flow on radial symmetric surface.  In addition to
this theoretical problem, we compare the numerical solution
of the S-SWE model with the analytic solution of
(\ref{e-swe}).  It is more difficult to numerical integrate
the E-SWE than the S-SWE model, and thus it is valuable to
know if the numerical solution of S-SWE can be still used as
a quantitative approximation of the real flow on a
non-planar surfaces.

\section{Effects of Variation of the Geometrical
  Characteristics of the Surface}
\label{section3}
Let us firstly point out that both models can be obtained
from Navier-Stokes equation by an asymptotic analysis.  The
S-SWE model can be alternatively obtained by considering the
surface base coordinate system to be a horizontal plane
$x^3=0$ or the base surface of the flow but assuming that
$\nu^3=1$.  In this sense, one can consider the S-SWE model
as simplified version of the E-SWE model.  Besides the two
mechanical quantities, water depth $h$ and water velocity
$\boldsymbol{v}$ describing the flow dynamics (they are also
the main unknowns of shallow water equations), there is
another important mechanical quantity: the energy of the
fluid.  This quantity obeys a conservative equation that is
derived from the primary equations (the mass conservation
and the linear momentum balance equations).  For simplicity,
we assume that there is no mass exchange and the frictional
force can be neglected, i.e. $\mathfrak{M}=0$ and
$\boldsymbol{\mathfrak{F}}=\boldsymbol{0}$.  With these
assumptions, the conservative equation of the energy
${\cal E}$ for the \mbox{E-SWE} model read as:
\begin{equation}
  \displaystyle\frac{\partial h\beta{\cal E}^{*}}{\partial t}+
  \partial_a\left(h\beta v^a {\cal E}\right)=0,
  \label{e_swe-energy}
\end{equation}
where
\begin{equation*}
  {\cal E}^{*}:=\frac{1}{2}|\boldsymbol{v}|^2+g(b^3+\frac{h}{2}\nu^3),\quad 
  {\cal E}:=\frac{1}{2}|\boldsymbol{v}|^2+g(b^3+h\nu^3).
\end{equation*}
Recall that $b^3(y^1,y^2)$ stands for the altitude of the
base flow surface.  The conservative equation of the energy
for the \mbox{S-SWE} model is given by (\ref{e_swe-energy})
with $\nu^3=1$.

In what follows, we analyze the flow on a radial symmetric
surface.  Such kind of surface admits a parametric
representation as:
\begin{equation}
  \left\{
    \begin{array}{l}
      x^1=r\cos{\theta}\\
      x^2=r\sin{\theta}\\
      x^3=f(r)
    \end{array}
  \right. ,
  \label{conic_surface}
\end{equation}
where $\theta\in[0,2\pi)$, $r\in [r_0,r_1]$.  We identify
the distance coordinate $r$ as the first component $y^1$,
and the polar angle $\theta$ as the second component $y^2$.

The unit normal to the surface is given by
\begin{equation}
  \boldsymbol{\nu} = \nu
  \left(
    \begin{array}{c}
      -f^{\prime}\cos{\theta}\\
      -f^{\prime}\sin{\theta}\\
      1\\
    \end{array}
  \right), \quad
  \nu=\displaystyle\frac{1}{\sqrt{1+(f^\prime)^2}}.
  \label{eq_unitary_normal}
\end{equation}
The covariant components of the metric and the elementary
area element are given by
\begin{equation}
  \beta_{\cdot\cdot}=
  \left(
    \begin{array}{cc}
      1+(f^\prime)^2&0\\
      0&r^2
    \end{array}
  \right), \quad
  \beta=r\sqrt{1+(f^\prime)^2}.
  \label{eq_beta}
\end{equation}
The components of the Christoffel symbol
$\gamma_{\cdot\cdot}^\cdot$ are given by
\begin{equation}
  \begin{array}{lll}
    \gamma^1_{11}=\displaystyle\frac{f^\prime f^{\prime\prime}}{\sqrt{1+(f^\prime)^2}},
    &\gamma^1_{12}=\gamma^1_{21}=0,
    &\gamma^1_{11}=\displaystyle\frac{-r}{1+(f^\prime)^2},\\
    \gamma^2_{11}=0,
    &\gamma^2_{12}= \gamma^2_{21}=\displaystyle\frac{1}{r},
    &\gamma^2_{22}=0.
  \end{array}
  \label{eq_ChristoffelSymbol}
\end{equation}

For such radial symmetric surface and for proper boundary
conditions, both models admit a 1-D type solution, $h$ and
$\boldsymbol{v}$, spatially depending only on the distance
$r$ and with $v^2=0$.

A very particular solution of the two models for the case of
a radial symmetric surface is the steady and axial symmetric
one,
$$v^1=u(r), \quad v^2=0, \quad h=h(r).$$
In the absence of mass source and frictional terms, the mass
conservation and energy conservation equations lead to
\begin{equation}
  \begin{split}
    \beta h u & =q_w,\\
    \frac{1}{2}|\boldsymbol{v}|^2+g\left(f(r)+h \nu\right) & =q_e,
  \end{split}
  \label{sol-a.1}
\end{equation}
where $q_w$ and $q_e$ are two constants, and $\nu$ and
$\beta$ are given by (\ref{eq_unitary_normal}) and
(\ref{eq_beta}), respectively. We remind the reader that
$\nu=1$ and $\beta=r$ for the S-SWE model.  Note that the
presence of the curvature in the E-SWE model is reflected by
the variation of $\nu$ with respect to $r$.  Taken into
account that
$$|\boldsymbol{v}|^2=\beta_{11}v^1v^1,$$
one can obtain an equation for $h$, the {\it h-profile
  equation}:
\begin{equation}
  g\nu h^3-(q_e-gf(r))h^2+\displaystyle\frac{q^2_w}{2r^2}=0.
  \label{h-eq-01}
\end{equation}
Solving the {\it h-profile equation} for each
$r\in[r_0,r_1]$, one obtains a solution
\begin{equation}
  \begin{array}{l}
    h=h(r; q_w, q_e),\\
    u=u(r;q_w, q_e),
  \end{array}
  \label{e-swe-sol}
\end{equation}
for each of the two models, where the two constants $q_w$
and $q_e$ are determined by the boundary data of
$\{h(r), u(r)\}$.

Although one can write an analytic expression for the
solution (\ref{e-swe-sol}), this would be difficult to be
handled for practical purposes.  This is the reason why it
is preferable to numerically solve for $h$ the {\it
  h-profile equation} as accurate as one wants.  We will
call it an analytic solution of the model and it will be the
base on which we will perform the following analysis.

Let us define the polynomials of the $h$-profile equations
for the two models (\ref{s-swe}) and (\ref{e-swe}),
respectively:
\begin{equation}
  {\cal P}^S_r (h) := g h^3 - (\tilde{q}_e-gf(r))h^2 + \displaystyle\frac{q_w^{\;2}}{2r^2},
  \label{eq_PS}
\end{equation}
\begin{equation}
  {\cal P}^E_r (h) := g\nu h^3 - (q_e-gf(r))h^2 + \displaystyle\frac{q_w^2}{2r^2},
  \label{eq_PE}
\end{equation}
where $q_w=rh|\boldsymbol{v}|$ and $\tilde{q}_e$ is defined
as in (\ref{sol-a.1}) for $\nu = 1$.  Note that each of the
above polynomials has one negative root.  For the case of
three real roots, the other two are positive.

For any fixed $r$, let us denote by $h^M(r)$ the water depth
solution of the two models, $M\in\{E,S\}$.  In what follows,
our goal is to obtain an estimation of how large is the
influence of $\nu$ on the difference between $h^S(r)$ and
$h^E(r)$.  We consider the case of a monotone decreasing
surface function $f(r)$ with $f(r_0)>f(r_1)$ and a boundary
data given at $r=r_0$
\begin{equation}
  h(r_0)=h_0, \quad u(r_0)=u_0.
  \label{eq_boundary}
\end{equation}
For the sake of simplicity, we consider that $\nu (r_0)=1$.

\begin{proposition}
  If $\displaystyle\frac{u_0^2}{gh_0}>1$ and
  $f^{\prime}(r) \leq 0$, then:
  \begin{itemize}
  \item[{\rm (1)}] the solution $h^M(r)$ of ${\cal P}^M_r(h)=0$
    exists for $M\in\{E,S\}$,
  \item[{\rm (2)}] $h^E(r) \leq h^S(r)$,
  \item[{\rm (3)}] $\exists c>0$
    s.t. $h^S(r)-h^E(r) \leq c\cdot (h^E(r))^2 (1-\nu)$,
  \item[{\rm (4)}]
    $h^S(r)-h^E(r) \geq L(r)\cdot (h^S(r))^2 (1-\nu)$, where
    \begin{equation*}
      L(r) = \frac{1}{3\nu(r)\left(\frac{2\eta(r)}{3g\nu(r)}-h^Q(r)\right)} > 0, \quad
      \eta(r) := q_e-gf(r), \quad
      h^Q(r) := \frac{q_w}{r\sqrt{2\eta(r)}},
    \end{equation*}
  \end{itemize}
  for all $r\in[r_0,r_1]$.
  \label{propozitie}
\end{proposition}
\begin{proof}
  (1) Knowing that
  \begin{equation*}
    \tilde{q}_e = q_e = \frac{1}{2}u_0^2+g\left(f(r_0)+h_0\right)
  \end{equation*}
  holds for $\nu(r_0)=1$
  and denoting
  \begin{equation*}
    \nu^M(r):=\left\{
      \begin{array}{lcl}
        \nu(r), &{\rm if} &M=E\\
        1, &{\rm if} &M=S
      \end{array}
    \right. ,
  \end{equation*}
  one observes the following:
  \begin{itemize}
  \item[(a)]
    $\displaystyle\partial_h{\cal P}^{M}_r (h) = 3 g\nu^M(r) h^2-2\eta(r)h,$
  \item[(b)]
    $\displaystyle{\cal P}^{M}_{r_0} (h_0) = 0, \quad\forall M\in\{E,S\}$,
  \item[(c)]
    $\partial_r{\cal P}^{S}_r (h)=gf^{\prime}(r)h^2-\displaystyle\frac{q_w^2}{r^3}$.
  \end{itemize}

  \noindent
  Using properties (a), (b), (c), one can now prove that
  each of the polynomial equation ${\cal P}^{M}_r (h) = 0$
  has three real roots and since
  \begin{equation*}
    \displaystyle \frac{2\eta(r_0)}{3g} - h_0 = \frac{1}{3g}(u_0^2-gh_0) > 0,
  \end{equation*}
  one concludes that only one root is a continuous function
  of $r$ which approaches $h_0$ as $r\rightarrow r_0$.  We
  call that root the solution of shallow water equation
  model $M$ and we denote it by $h^M(r)$, where
  $M\in\{E,S\}$.

  \noindent
  (2) For each $M\in\{E,S\}$, let us denote by
  $h_{\star}^M(r)$ the nonzero root of the derivative of
  ${\cal P}_r^M$ with respect to $h$, i.e.
  \begin{equation*}
    h_{\star}^M(r) := \displaystyle\frac{2}{3g}\frac{\eta(r)}{\nu^M(r)}.
  \end{equation*}

  Observing that
  \begin{equation*}
    \begin{array}{c}
      h^S_{\star}(r) > h^S_{\star}(r_0) > h_0,\\
      \partial_r{\cal P}^{S}_r (h)<0,\quad {\cal P}^{E}(h)<{\cal P}^{S}(h), \quad \forall h>0,
    \end{array}
  \end{equation*}
  one can immediately conclude that
  \begin{equation*}
    h^E(r)<h^S(r)<h_0, \quad r\in(r_0,r_1),
  \end{equation*}
  see also Fig.~\ref{fig_1}.

  \noindent
  (3) Since
  \begin{equation*}
    {\cal P}^S (h^E) = {\cal P}^S (h^E) - {\cal P}^E (h^E) = g(1-\nu)(h^E)^3
  \end{equation*}
  and
  \begin{equation*}
    {\cal P}^S (h^E) = {\cal P}^S (h^E) - {\cal P}^S (h^S) = ({\cal P}^S)^\prime(\xi)(h^E-h^S),
  \end{equation*} 
  for some $\xi\in (h^E,h^S)$, then
  \begin{equation*}
    h^S-h^E = \displaystyle\frac{(1-\nu)(h^E)^3}{3\xi(h^S_{\star}-\xi)}.
  \end{equation*}
  As $h^E<\xi<h^S$ and
  $h^S_{\star}(r)-\xi>h^S_{\star}(r_0)-h_0$, one can now
  easily conclude that
  \begin{equation}
    h^S-h^E<(1-\nu)(h^E)^2\displaystyle\frac{g}{u_0^2-g h_0}.
    \label{eq_hS_hE_uppbound}
  \end{equation}

  \noindent
  (4) Observe that
  \begin{equation*}
    \eta(r) >0, \quad \forall r\in [r_0,r_1],
  \end{equation*}
  $h^Q(r)$ is the positive root of the polynomial
  \begin{equation}
    {\cal Q}_r(h):=-\eta(r)h^2+\frac{q_w^2}{2r^2},
    \label{eq_Qr}
  \end{equation}
  and that
  \begin{equation}
    {\cal P}_r^S(h) \geq {\cal P}_r^E(h) \geq {\cal Q}_r(h), \quad \forall h\geq 0.
    \label{ineg_PPQ}
  \end{equation}
  But ${\cal P}_r^S(\cdot)$, ${\cal P}_r^E(\cdot)$,
  ${\cal Q}_r(\cdot)$ are monotone decreasing on
  $\left [0,\frac{2\eta(r)}{3g}\right]$ and
  ${\cal Q}_r(h^Q) = {\cal P}_r^E(h^E) = {\cal P}_r^S(h^S) = 0$;
  thus (\ref{ineg_PPQ}) now gives
  \begin{equation}
    0 < h^Q(r) < h^E(r) \leq h^S(r) \leq \frac{2\eta(r)}{3g}, \quad \forall r\in[r_0,r_1].
    \label{bound_hS}
  \end{equation}
  Since
  \begin{equation*}
    {\cal P}^E (h^S) = {\cal P}^E (h^S) - {\cal P}^S (h^S) = g(\nu-1)(h^S)^3
  \end{equation*}
  and
  \begin{equation*}
    {\cal P}^E (h^S) = {\cal P}^E (h^S) - {\cal P}^E (h^E) = ({\cal P}^E)^\prime(\xi)(h^S-h^E),
  \end{equation*} 
  for some $\xi\in (h^E,h^S)$, then
  \begin{equation}
    h^S-h^E = \displaystyle\frac{(1-\nu)(h^S)^3}{3\nu\xi(h^E_{\star}-\xi)}.
    \label{eq_hS_hE_lowbound}
  \end{equation}
  The ``lower bound'' inequality from the proposition can
  now be easily obtained observing that $0<\xi<h^S$ and
  $0 < h^E_{\star}-\xi < h^E_{\star}-h^Q$.
\end{proof}

\begin{figure}[htbp]
  \centering
  \includegraphics[width=0.63\linewidth]{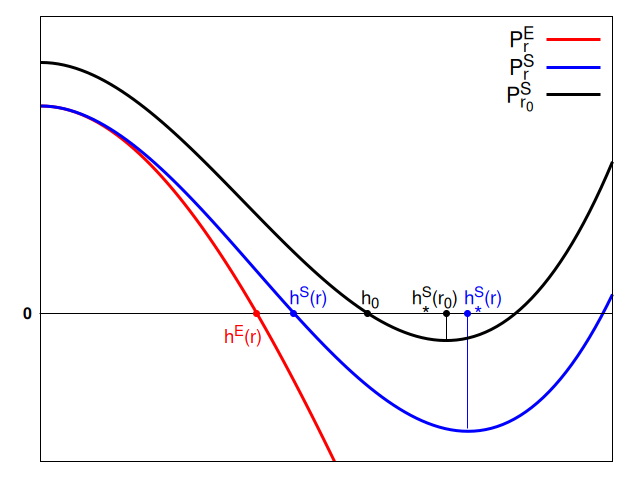}
  \caption{Graphical sketch for
    Proposition~\ref{propozitie}. The black curve is the
    graph of the polynomial function ${\cal P}^{E}_{r_0}$
    which coincides with ${\cal P}^{S}_{r_0}$.  The red and
    the blue curves are the graphs of ${\cal P}^{E}_{r}$ and
    ${\cal P}^{S}_{r}$ for a fixed $r$, respectively.}
  \label{fig_1}
\end{figure}

\begin{remark}
  If $\nu(r_0)<1$ then statement {\rm (1)} will continue to
  hold, statement {\rm (2)} may not be true along the entire
  interval $(r_0, r_1)$, and statement {\rm (3)} must be
  read as
  $$\left| h^S(r)-h^E(r) \right| \leq c \cdot (h^E(r))^2 (1-\nu).$$
\end{remark}
\begin{remark}
  For the case of monotone increasing profile $f^{\prime}>0$
  with boundary data given at $r=r_1$ and $\nu(r_1)=1$,
  statement {\rm (4)} continues to hold provided that the
  solution $h^S$ exists for any $r\in(r_0,r_1)$.
  Furthermore, if in addition $h^S$ remains upper bounded by
  $h_0$, then statement {\rm (3)} also continues to hold.
\end{remark}
In order to ``illustrate'' Proposition~\ref{propozitie} and
the remarks, we consider two particular cases of profiles
$f$.  Fig.~\ref{fig_upperbound} shows the graphs of the relative
error \mbox{$E_r:={\left(h^S-h^E\right)}/{h^E}$} and
\begin{equation}
  U_b:=\frac{h^E}{h_0}\frac{1-\nu}{Fr^2 -1}
  \label{eq_ub}
\end{equation}
that can be derived from (\ref{eq_hS_hE_uppbound}), where
$Fr:=u_0/{\sqrt{gh_0}}$ is the Froude number.  For the
conditions in Proposition~\ref{propozitie}, $U_b$ is an
upper bound for $E_r$.
\begin{figure}[htbp]
  \centering
  \begin{tabular}{cc}
    \includegraphics[width=0.38\linewidth]{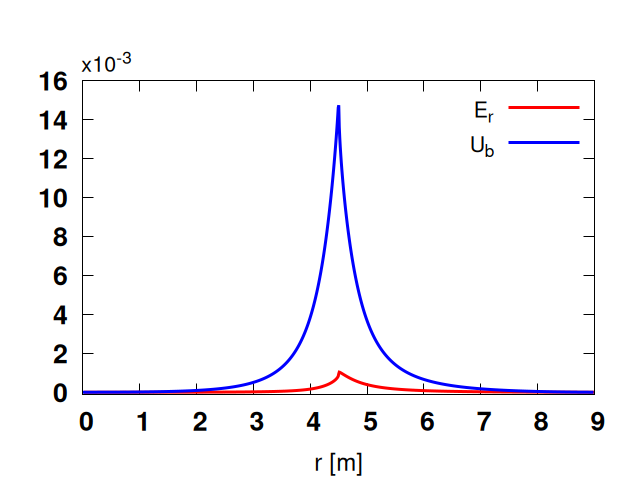}
    &\includegraphics[width=0.38\linewidth]{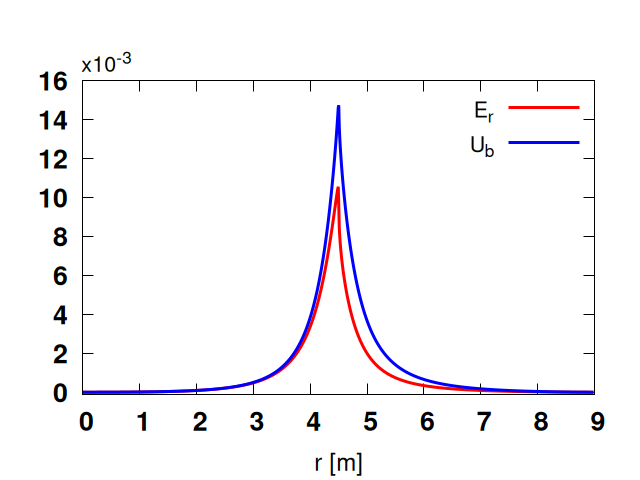}\\
    \includegraphics[width=0.38\linewidth]{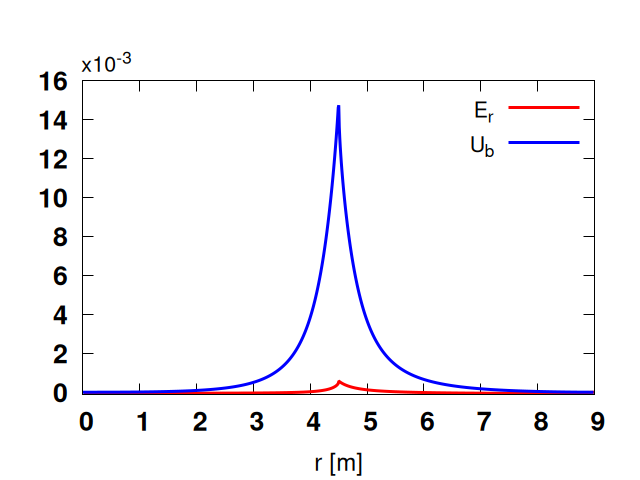}
    &\includegraphics[width=0.38\linewidth]{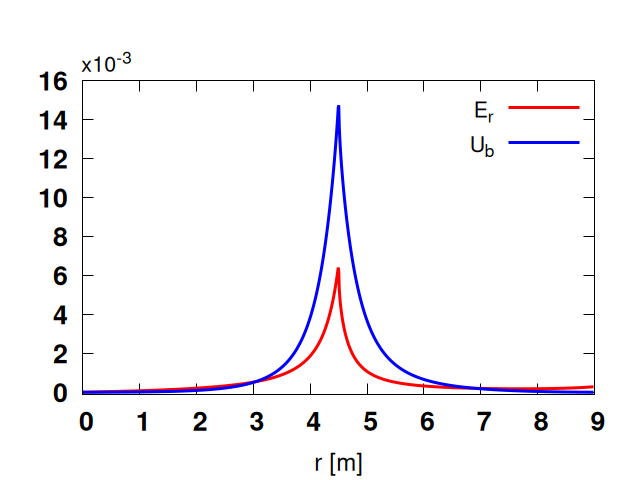}
  \end{tabular}
  \caption{Relative errors between $h^E$ and $h^S$ for
    surfaces generated by two different profiles $f$ from
    Fig.~\ref{fig_surfaces}.  The first column corresponds
    to a water flow on a type (III) surface, while the
    second column corresponds to a flow on a type (II)
    surface.  The profile $f$ has $\nu_0 = 1$ on the first
    row and $\nu_0 < 1$ on the second row.}
  \label{fig_upperbound}
\end{figure}

\noindent
Note that the hypotheses of Proposition~\ref{propozitie} are
fulfilled here only by the flow corresponding to the top
left figure.

\begin{remark}
  The estimation given by statement {\rm (3)} is mainly
  important for practical applications where the water depth
  has small values and it shows that the simplified model
  S-SWE is still able to reflect the influence of variation
  in the curvature of the base flow surface on the dynamics
  of water flow.
\end{remark}

\begin{remark}
  Although the difference between the solutions of the two
  models is generally small, the estimation given by
  statement {\rm (4)} emphasizes that this difference exists
  and it is not zero.
\end{remark}
Fig.~\ref{fig_lowerbound} reflects a particular case where
the difference between the two solutions $h^S$ and $h^E$ is
not so neglectable anymore.
\begin{figure}[htbp]
  \centering
  \begin{tabular}{cc}
    \includegraphics[width=0.41\linewidth]{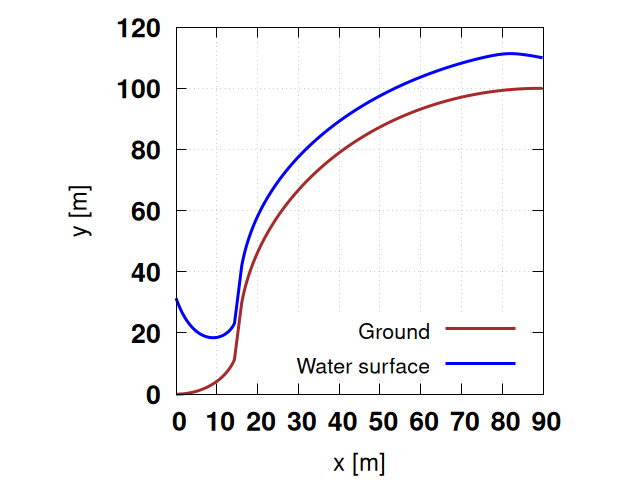}
    &\includegraphics[width=0.41\linewidth]{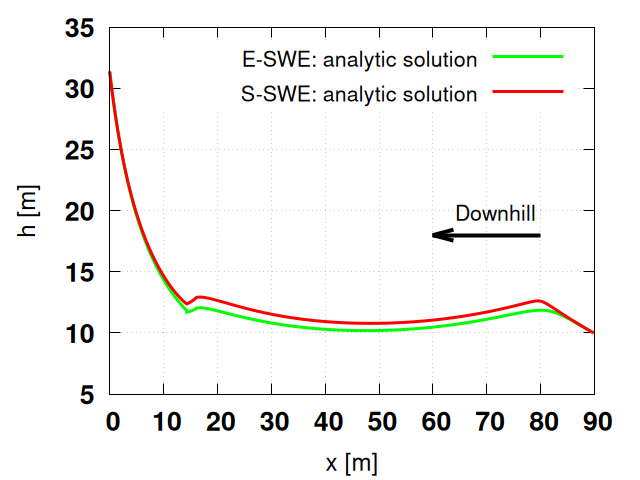}\\
    \includegraphics[width=0.41\linewidth]{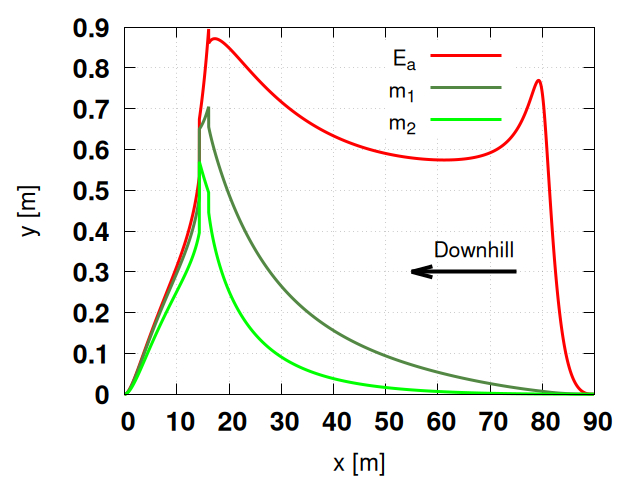}
    &\includegraphics[width=0.41\linewidth]{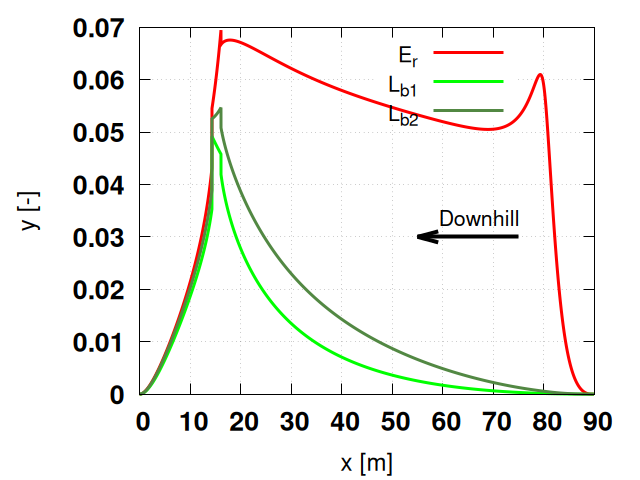}
  \end{tabular}
  \caption{The stationary water flow on a radial symmetric
    surface generated by two quarter circles.  The top left
    picture presents the profile $f$ with
    \mbox{$\nu(r_1)=1$} and the water surface of the steady
    state flow.  The top right picture illustrates the
    graphics of the water depth corresponding to the
    solutions of the two models.  The absolute and relative
    errors (between $h^E$ and $h^S$) and their lower bounds
    defined in (\ref{eq_lb_abs}) and (\ref{eq_lb_rel}) are
    pictured on the bottom row.}
  \label{fig_lowerbound}
\end{figure}

\noindent
It shows the graphs of
\begin{equation}
  E_a:=h^S-h^E, \quad
  m_{b1}:=L\cdot (h^S)^2 (1-\nu), \quad
  m_{b2}:=L\cdot (h^Q)^2 (1-\nu), \quad
  \label{eq_lb_abs}
\end{equation}
\begin{equation}
  E_r:={\left(h^S-h^E\right)}/{h^S}, \quad
  L_{b1}:=\frac{m_{b1}}{h^S}, \quad
  L_{b2}:=\frac{m_{b2}}{h^Q}.
  \label{eq_lb_rel}
\end{equation}
that can be derived from the last property of
Proposition~\ref{propozitie}.  For the conditions in
Proposition~\ref{propozitie}, $m_b$ and $L_b$ are lower
bounds for the absolute $E_a$ and relative $E_r$ errors,
respectively.

\section{Numerical Solutions}
For most real life problems, one does not know the analytic
solutions of the shallow water models as in
\cite{delestre,matskevich}, and as a consequence, one uses
the numerical solutions of S-SWE
\cite{COZZOLINO201883,Hsu2019,cell_dray,seguin,fayssal,nordic,liu,marche}.
The numerical methods developed to solve S-SWE have an
intrinsic mathematical interest and they provide solutions
for real problems coming from many other areas
\cite{virgil-conceptual, navaro, BRESCH20091}.  Taking this
into account, it is relevant (from the application point of
view) to compare the numerical solutions of S-SWE with the
exact ones of E-SWE.

As we pointed out in the previous sections, it is of main
interest to investigate the effects of the curvature
variation on the solution of shallow water equation.  We
noticed that both models give basically the same solutions
for a large class of surfaces.  However, one can find
certain surfaces leading to significant differences between
the solutions of the two models.

In many applications, especially in hydrological sciences,
one considers that the average inclination of the base flow
surface is the dominant factor that determines the water
dynamics.  But in real problems, the surface exhibits local
variation of its geometrical properties.  To illustrate the
influence of such local variation of the surface curvature,
we consider four types of surfaces as in
Fig.~\ref{fig_surfaces}, all having the same average slope.
Three surfaces have radial symmetry and the fourth one is
generated by the translation of a 1D profile.
\begin{figure}[ht!]
  \centering
  \begin{tabular}{cc}
    \includegraphics[width=0.49\linewidth]{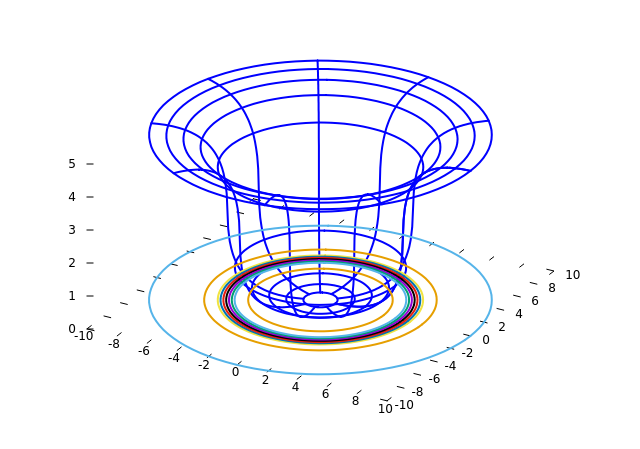}
    &\includegraphics[width=0.49\linewidth]{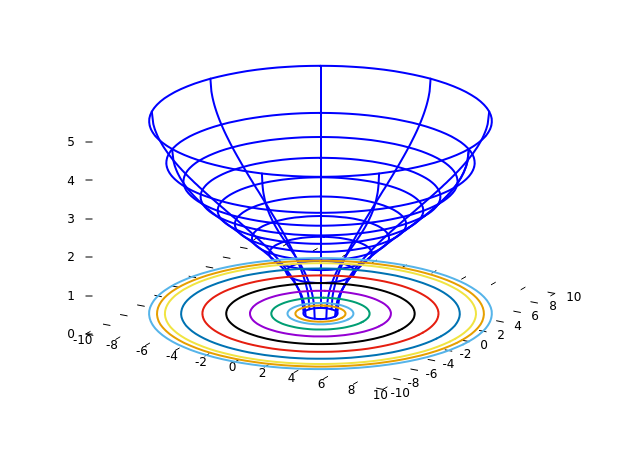}\\
    Crater type (I) &Crater type (II)\\
    \includegraphics[width=0.49\linewidth]{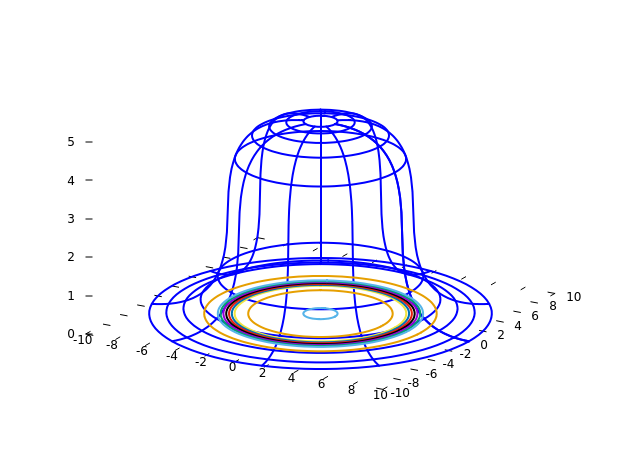}
    &\includegraphics[width=0.49\linewidth]{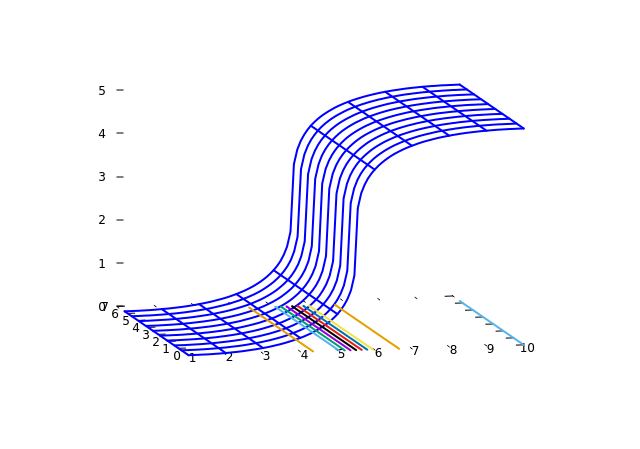}\\
    Hillock type (III) &Waterfall $1$D type (IV)
  \end{tabular}
  \caption{Four different types of surfaces $f$. All
    surfaces have the same height, $5$ m.  Surfaces I, III
    and IV are generated by functions of type
    $g(\xi)=\sqrt[3]{\xi}-\xi/3$, while surface II is
    generated by an $\arcsin$ function, applying convenient
    translations and dilatation of arguments. The three
    radial symmetric surfaces are defined for $r\in[1,10]$
    and the fourth surface for $x\in[1,10]$.}
  \label{fig_surfaces}
\end{figure}

The boundary conditions are the same for all four examples:
$h_0 = 0.1$ m and $Fr=5$ at the top of the surface and free
drainage the bottom.

The two left pictures of Fig.~\ref{fig_efectele_curburii}
(built on similar surfaces with the same line curvature with
respect to $\theta$) show the influence of the radial
curvature on the flow dynamics on crater type surfaces (I
and II) where water accumulates at the bottom.  We note that
the $h$ profiles for both surfaces remain close to each
other, while the $v$ profiles have significant differences.
The two right pictures of Fig.~\ref{fig_efectele_curburii}
show the influence of the $\theta$ line curvature on the
flow.  As opposed to the left, the $v$ profiles for all
three types of soil surfaces (I, III and IV) stay close to
each other, but there are considerable differences in the
$h$ profiles.
\begin{figure}[htbp]
  \centering
  \begin{tabular}{cc}
    \includegraphics[width=0.38\linewidth]{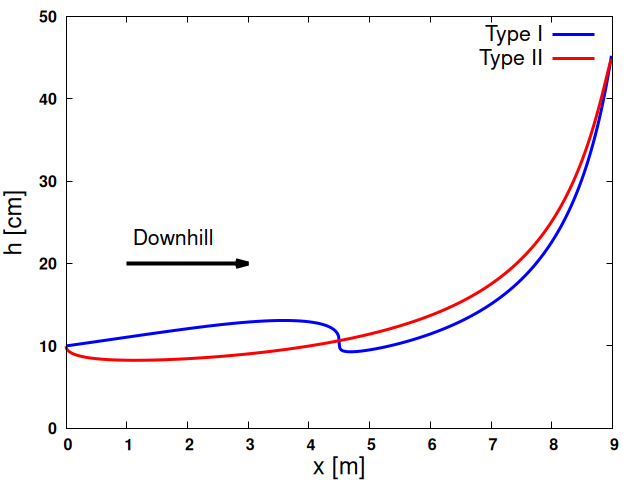}
    &\includegraphics[width=0.38\linewidth]{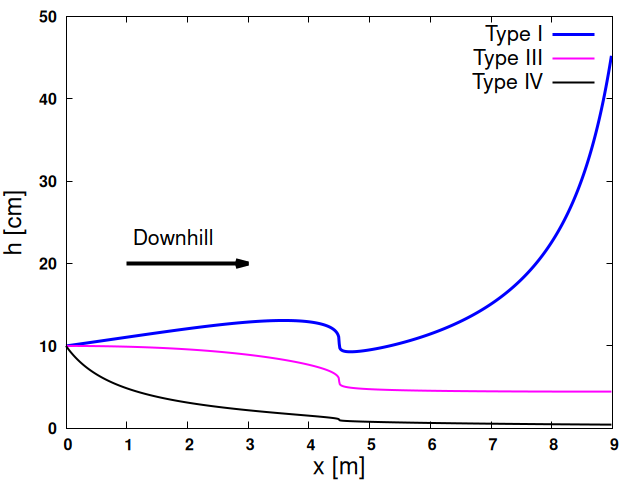}\\
    \includegraphics[width=0.38\linewidth]{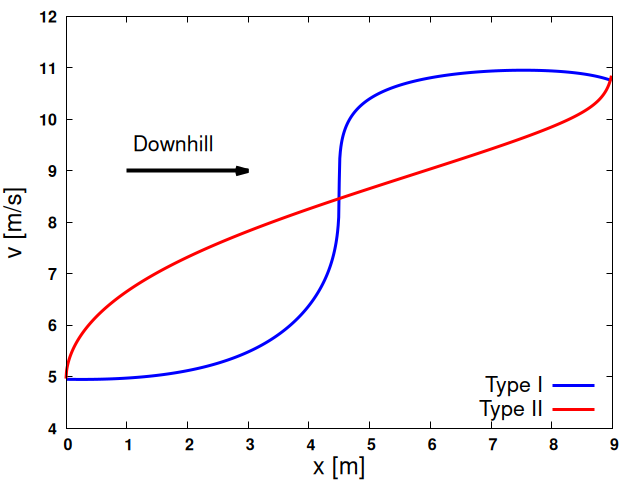}
    &\includegraphics[width=0.38\linewidth]{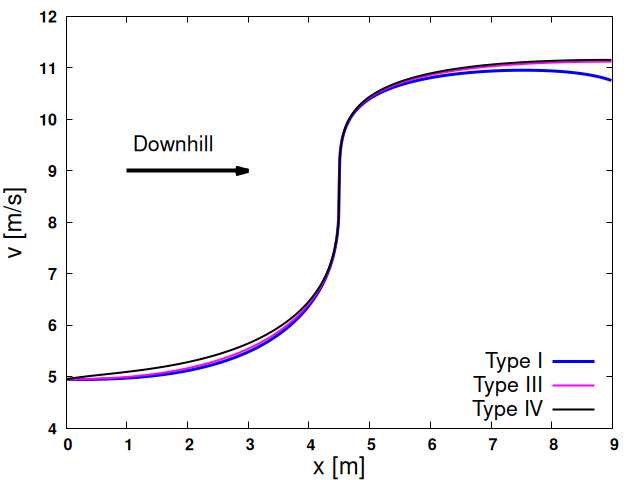}\\
  \end{tabular}
  \caption{Effects of the curvature on the flow distribution
    along four types of surfaces $f$.}
  \label{fig_efectele_curburii}
\end{figure}

Fig.~\ref{fig_sol_num_analytic} shows there is a good
agreement between the analytic and numerical solutions of
the two models for two types of surfaces $f$ (I and II).
\begin{figure}[htbp]
  \centering
  \begin{tabular}{cc}
    \includegraphics[width=0.4\linewidth]{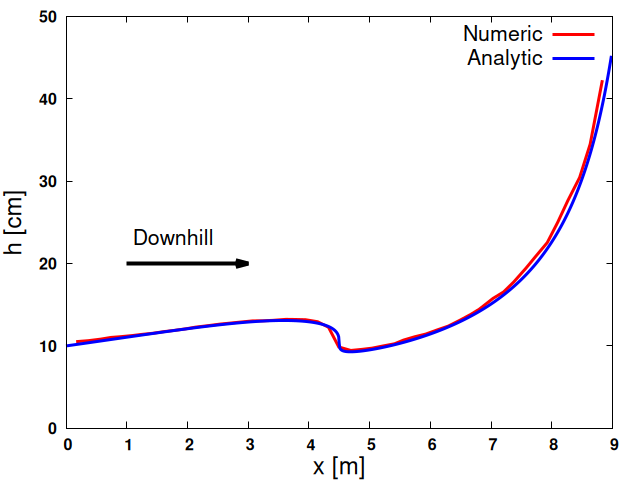}
    &\includegraphics[width=0.4\linewidth]{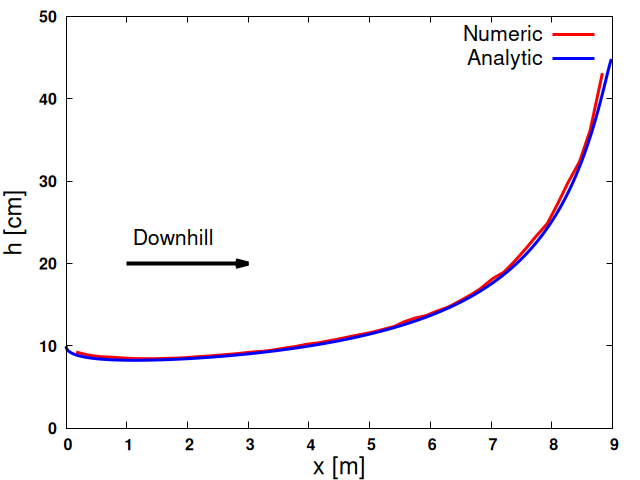}\\
    \includegraphics[width=0.4\linewidth]{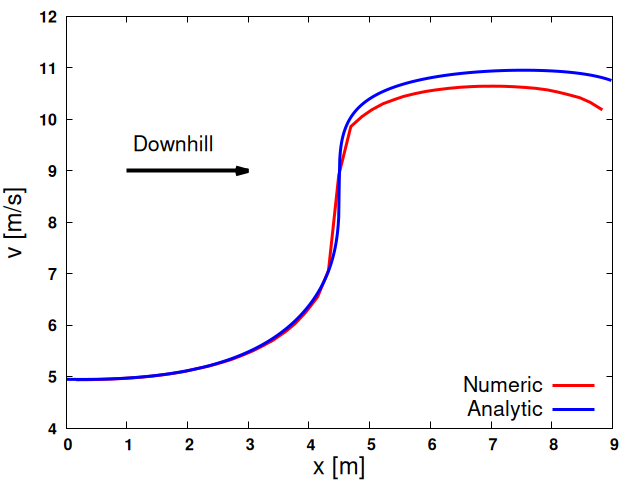}
    &\includegraphics[width=0.4\linewidth]{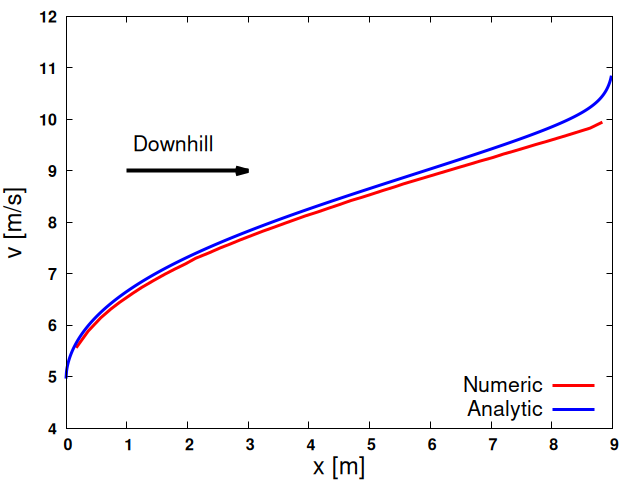}
  \end{tabular}
  \caption{Comparison between the analytic solution of E-SWE
    and the numerical solution of S-SWE for the water flow
    on surface I (left column) and water flow on surface II
    (right column).}
  \label{fig_sol_num_analytic}
\end{figure}
The numerical solution of S-SWE (\ref{s-swe}) is obtained by
applying the method described in \cite{sds-apnum}.  The
finite volume method uses a triangular mesh (obtained by
using a quality mesh generator (see \cite{shewchuk96b,
  shewchuk} for details)) on a 2D domain.  In all the
applications with radial symmetric surfaces (I, II and III)
we consider here, the 2D domain is an annulus
\begin{equation*}
  D:=\{(r,\theta)\left| \, r\in[r_0,r_1], \, \theta\in[0,2\pi)\right.\}.
\end{equation*}
and the boundary conditions are of Dirichlet type for the
upper part and free discharge for the bottom part of the
surface:
\begin{equation}
  h(t,\boldsymbol{x})=h_0=0.1, \quad 
  \boldsymbol{v}(t,\boldsymbol{x})=-v_0\boldsymbol{n},
  \quad \forall\boldsymbol{x}\in\partial D^{\rm in},
  \label{numerics-bc}
\end{equation}
where $\partial D^{\rm in}$ is the upper boundary of the
surface, $\boldsymbol{n}$ is the external oriented unit
normal vector of $\partial D^{\rm in}$ at $\boldsymbol{x}$,
and \mbox{$v_0=5\sqrt{0.1g}$}.

\section{Comments and Conclusions}
The E-SWE model can be seen as an intermediate model between
the S-SWE model and a mode general model where the
hydrostatic approximation of the pressure field is not
enforced.  The E-SWE model tries to be more realistic by
taking into account the variation of the normal direction to
the surface.  From the analysis presented in this article,
we can conjecture that {\it for a moderate variation of the
  surface curvature, the gradient of the surface is
  sufficient to find solutions that catch the variation of
  surface curvature}, see Fig.~\ref{fig_efectele_curburii}.

We can also make a remark concerning the physical
signification of the shallow water S-SWE model: the good
agreement between the solutions of the two models (see
Fig.~\ref{fig_upperbound} and the remarks from
Section~\ref{section3}) suggest us to think the velocity
$\boldsymbol{v}$ of the S-SWE model as vector in the tangent
plane to the surface and $h$ as water depth in the normal
direction to the surface.

Finally, note that for an adequate utilization of the
\mbox{S-SWE} model in real life problems, one must specify a
proper reference surface for the variables.



\section*{Acknowledgment}
Some of this work was partially supported by a grant of
Romanian Ministry of Research and Innovation,
CCCDI-UEFISCDI, project number
PN-III-P1-1.2-PCCDI-2017-0721/34PCCDI/2018 within PNCDI III.

\bibliographystyle{elsarticle-num}
\bibliography{eco_SWE}

\end{document}